\documentclass[fleqn,12pt,twoside]{article}

\usepackage{graphicx}

\input{layout_commands}
\usepackage{amsmath}
\usepackage{amssymb}

\newcounter{mynumcounter}
	       {%
		 \begin{list}{(\roman{mynumcounter})\hspace*{\fill}}%
		   {
		     \setlength{\topsep}{0cm}
		     \setlength{\partopsep}{0cm}
		     \setlength{\itemsep}{0ex}
		     \setlength{\parsep}{0cm}
		     \setlength{\leftmargin}{0cm}
		     \setlength{\itemindent}{8mm}		   
		     \setlength{\labelsep}{3mm}
		     \setlength{\labelwidth}{5mm}
		     \usecounter{mynumcounter}
		   }%
	       }
	       {\end{list}}

\newcommand{\eps}{\varepsilon}
\newcommand{\del}{\partial}

\newcommand{\dd}[2]{\frac{\del #1}{\del #2}}
\newcommand{\ddeval}[3]{\left.\dd{#1}{#2}\right|_{#3}}

\newcommand{\tr}{\operatorname{tr}}
\renewcommand{\div}{\operatorname{div}}

\newcommand{\IR}{\mathbf{R}}

\newcommand{\CH}{\mathcal{H}}

\newcommand{\CJ}{\mathcal{J}}

\newcommand{\CO}{\mathcal{O}}
\newcommand{\CP}{\mathcal{P}}

\newcommand{\id}{\operatorname{id}}

\newcommand{\rmd}{\mathrm{d}}

\newcommand{\dist}{\mathrm{dist}}

\newcommand{\dmu}{\,\rmd\mu}

\newcommand{\ra}{\rangle}
\newcommand{\la}{\langle}

\newcommand{\Connection}[1]{\smash{\sideset{^{#1}}{}{\mathop\nabla\nolimits}}}
\newcommand{\Nabla}[1]{\Connection{#1}}

\newcommand{\Divergence}[1]{\smash{\sideset{^{#1}}{}{\mathop\mathrm{div}\nolimits}}}

\newcommand{\divM}{\!\Divergence{M}}

\newcommand{\Scalarcurv}[1]{\smash{\sideset{^{#1}}{}{\mathop\mathrm{Sc}\nolimits}}}

\newcommand{\ScalM}{\!\Scalarcurv{M}}
\newcommand{\ScalSig}{\!\Scalarcurv{\Sigma}}

\newcommand{\Acircbar}%
{\hspace*{4pt}\raisebox{8.5pt}{\makebox[-4pt][l]{$\scriptstyle\circ$}}\bar A}
\newcommand{\Kcircbar}%
{\hspace*{4pt}\raisebox{8.5pt}{\makebox[-4pt][l]{$\scriptstyle\circ$}}\bar K}

\newcommand{\half}{\tfrac{1}{2}}

\input{draft_commands}

\newcommand{\graph}{\operatorname{graph}}

\title{\titlefamily
Blowup of Jang's equation at outermost marginally trapped surfaces
}
\author{
  \authname{Jan Metzger
  \thanks{Research on this project started while the author was supported in part by a Feodor-Lynen Fellowship of the Humboldt Foundation.}}\\
  \authaddress{jan.metzger@aei.mpg.de\\[.2ex]
    Albert-Einstein-Institut,
    Am M\"uhlenberg 1,
    D-14476 Potsdam,
    Germany.}
}
\date{}

\begin{document}
\hyphenation{}
\pagestyle{footnumber}
\maketitle
\thispagestyle{footnumber}
\begin{abst}%
  The aim of this paper is to accurately describe the blowup of
  Jang's equation. First, we discuss how to construct solutions that
  blow up at an outermost MOTS. Second, we exclude the possibility
  that there are extra blowup surfaces in data sets with non-positive
  mean curvature. Then we investigate the rate of convergence of the
  blowup to a cylinder near a strictly stable MOTS and show
  exponential convergence with an identifiable rate near a strictly
  stable MOTS.
\end{abst}
\section{Introduction}
\label{sec:introduction}
This paper is concerned with the examination of the relation of Jang's
equation to marginally outer trapped surfaces (MOTS). To set the perspective,
we consider Cauchy data $(M,g,K)$ for the Einstein equations. Such
data sets are 3-manifolds $M$ equipped with a Riemannian metric
together with a symmetric bilinear form $K$ representing the second
fundamental form of the time slice $M$ in space-time. A marginally
outer trapped surface is a surface with $\theta^+ = H + P = 0$, where
$H$ is the mean curvature of $\Sigma$ in $M$ and $P = \tr K -
K(\nu,\nu)$ for the normal $\nu$ to $\Sigma$.

In the paper \cite{Andersson-Metzger:2007}, inspired by an idea of
Schoen \cite{Schoen:2004}, we constructed MOTS in the presence of
barrier surfaces by inducing a blow-up of Jang's equation. In this
context, Jang's equation~\cite{Schoen-Yau:1981,Jang:1978} is an
equation of prescribed mean curvature for the graph of a function in
$M\times\IR$. For details we refer to
section~\ref{sec:preliminaries}. 

In this note, we take a slightly different perspective. Consider a
data set $(M,g,K)$ with a non-empty outer boundary $\del^+M$ and
assume that we are given the outermost MOTS $\Sigma$ in
$(M,g,K)$. Here, \emph{outermost} means that there is no other MOTS on
the outside of $\Sigma$. From \cite{Andersson-Metzger:2007} it follows
that $(M,g,K)$ always contains a unique such surface, or does not
contain outer trapped surfaces at all, under the assumption that $\del
M$ is outer untrapped. As stated in theorem~\ref{thm:blowup}, we show
that there exists a solution $f$ to Jang's equation that actually
blows up at $\Sigma$, assuming that $\del M$ is inner and outer
untrapped. By blow-up we mean, that outside from $\Sigma$ the function
$f$ is such that $\graph f$ is a smooth submanifold of $M\times\IR$
with a cylindrical end converging to $\Sigma\times\IR$. There is
however a catch, as $f$ may blow up at other surfaces, too.  These
surfaces are marginally inner trapped. In theorem~\ref{thm:nonpos_mc}
we show that the other blow-up surfaces can not occur if the data set
has non-positive mean curvature. 

To put the result in perspective note that if the dominant energy
condition holds, the graph of $f$ is of non-negative Yamabe type and
thus can be equipped with a (singular) metric of zero scalar
curvature. This was used by Schoen and Yau in \cite{Schoen-Yau:1981}
to prove the positive mass theorem. Later Bray \cite{Bray:2001}
proposed to use Jang's equation to relate the Penrose conjecture in
its general setting to the Riemannian Penrose inequality
\cite{Huisken-Ilmanen:2001,Bray:2001} on a manifold related to Jang's
graph. One of the main questions in this program is whether or not
Jang's equation can be made to blow up at a specific MOTS. This
question was raised in the literature, cf. for example
\cite{Malec-OMurchadha:2004} where this question is discussed in the
rotationally symmetric case. Here we give the positive answer that
blow-up solutions exist at outermost MOTS. The author recently learned
that the existence of the blow-up solution is used in
\cite{khuri:2009} to prove a Penrose-like inequality.

With the blowup constructed, we can turn to the asymptotic behavior
of the blowup itself. It has been shown in~\cite{Schoen-Yau:1981} that
such a blowup must be asymptotic to a cylinder over the outermost
MOTS. In theorems~\ref{thm:exponential} and~\ref{thm:rigidity} we show
that under the assumption of strict stability the convergence rate is
exponential with a power directly related to the principal eigenvalue
of the MOTS. The general idea is to show the existence of a super-solution
with at most logarithmic blowup of the desired rate. Turning the
picture sideways yields exponential decay, when writing the solution
as a graph over the cylinder in question. Furthermore, we show that
beyond a certain decay rate, the solution must be trivial, thus
exhibiting the actual rate.

We expect that the knowledge of these asymptotics is tied to the
question whether the blow-up solution is unique. Furthermore note that
the constant in the Penrose-like inequality in~\cite{khuri:2009}
depends on the geometry of the solution. We thus expect that the value
on this constant is related to the asymptotic behavior near the
blow-up cylinder.

Before turning to these results, we introduce some notation in
section~\ref{sec:preliminaries}. Section~\ref{sec:blowup} proceeds
with the construction of the the blow-up. We will not go into details
here, but emphasize the general idea and point to the results needed
from the paper \cite{Andersson-Metzger:2007}. In
section~\ref{sec:asymptotic-behavior}, we perform the calculation of
the asymptotics.


%
\section{Preliminaries}
\label{sec:preliminaries}
Let $(M,g,K)$ be an initial data set for the Einstein equations. That
is $M$ is a 3-dimensional manifold, $g$ a Riemannian metric on $M$ and
$K$ a symmetric 2-tensor. We do not require any energy condition to hold.

Assume that $\del M$ is the disjoint union $\del M = \del^- M \cup
\del^+ M$, where $\del^\pm M$ are smooth surfaces without boundary. We
refer to $\del^- M$ as the inner boundary and endow it with the normal
vector $\nu$ pointing into $M$. The outer boundary $\del^+ M$ is
endowed with the normal $\nu$ pointing out of $M$. We denote by
$H[\del M]$ the mean curvature of $\del M$ with respect to the normal
vector field $\nu$, and by $P[\del M]=\tr_{\del M} K$ the trace of the
tensor $K$ restricted to the 2-dimensional surface $\del M$. Then the
inward and outward expansions of $\del M$ are defined by
\begin{equation*}
  \theta^\pm[\del M] = P[\del M] \pm H[\del M].
\end{equation*}
Assume that $\theta^+[\del^- M] = 0$, and that $\theta^+[\del^+ M]>0$
and $\theta^-[\del^+M]<0$.

If $\Sigma\subset M$ is a smooth, embedded surface homologous to
$\del^+ M$, then $\Sigma$ bounds a region $\Omega$ together with
$\del^+ M$. In this case, we define $\theta^\pm[\Sigma]$ as above,
where $H$ is computed with respect to the normal vector field pointing
into $\Omega$ (that is in direction of $\del^+M$). $\Sigma$ is called
marginally outer trapped surface (MOTS), if $\theta^+[\Sigma]=0$.  We
say that $\del M$ is an outermost MOTS, if there is no other MOTS in
$M$, which is homologous to $\del^+ M$. In
\cite{Andersson-Metzger:2007} it is proved that for any initial data
set $(M,g,K)$ which contains a MOTS, there is also an outermost MOTS
surrounding it.

Let $\Sigma\subset M$ be a MOTS and consider a normal variation of
$\Sigma$ in $M$, that is a map $F: \Sigma \times (-\eps,\eps) \to M$
such that $F(\cdot,0)=\id_\Sigma$ and $\ddeval{}{s}{s=0} F(p,s) =
f\nu$, where $f$ is a function on $\Sigma$ and $\nu$ is the normal of
$\Sigma$. Then the change of $\theta^+$ is given by
\begin{equation*}
  \ddeval{\theta^+[F(\Sigma,s)]}{s}{s=0} = L_M f,
\end{equation*}
where $L_M$ is a quasi-linear elliptic operator of second order along
$\Sigma$. It is given by
\begin{equation*}
  L_M f =  -\Delta f + 2 S(\nabla f) + f\big(\div S -|\chi^+|^2 - |S|^2
  + \tfrac{1}{2} \ScalSig -\mu - J(\nu)\big).
\end{equation*}
In this expression $\nabla$, $\div$ and $\Delta$ denote the gradient,
divergence and Laplace-Beltrami operator tangential to $\Sigma$. The
tangential 1-form $S$ is given by $S = K(\cdot,\nu)^T$, $\chi^+$ is
the bilinear form $\chi^+ = A + K^\Sigma$, where $A$ is the second
fundamental form of $\Sigma$ in $M$ and $K^\Sigma$ is the projection
of $K$ to $T\Sigma\times T\Sigma$. Furthermore, $\ScalSig$ denotes the
scalar curvature of $\Sigma$, $\mu = \half(\ScalM -|K|^2 + (\tr
K)^2)$, and $J = \divM K - d \tr K$. For a more detailed investigation
of this operator we refer to \cite{Andersson-Mars-Simon:2005} and
\cite{andersson-mars-simon:2007}.

The facts we will need here is that $L_M$ has a principal eigenvalue
$\lambda$, which is real and has a one-dimensional eigenspace which is
spanned by a positive function. If $\lambda$ is non-negative $\Sigma$
is called stable, and if $\lambda$ is positive, $\Sigma$ is called
strictly stable.  In particular, if $\Sigma$ is strictly stable as a
MOTS, there exists an outward deformation strictly increasing
$\theta^+$.

In $\bar M = M \times \IR$, we consider Jang's equation
\cite{Jang:1978,Schoen-Yau:1981} for the graph of a function $f: M\to 
\IR$. Let $N := \graph f = \{(x,z) : z=f(x)\}$. The mean curvature
$\CH[f]$ of $N$ with respect to the downward normal is given by
\begin{equation*}
  \CH[f] = \div \left( \frac{\nabla f}{\sqrt{1 + |\nabla f|^2}} \right)
\end{equation*}
define $\bar K$ on $\bar M$ by $\bar K_{(x,z)} (X,Y) = K_x(\pi X,\pi
Y)$, where $\pi: T\bar M \to TM$ denotes the orthogonal projection
onto the horizontal tangent vectors. Let
\begin{equation*}
  \CP[f] = \tr_N \bar K.  
\end{equation*}
Then Jang's equation becomes
\begin{equation}
  \label{eq:1}
  \CJ[f] = \CH[f] - \CP[f] = 0.
\end{equation}


%
\section{The blowup}
\label{sec:blowup}
The main result of this section is that we can construct a solution to
Jang's equation which blows up at the outermost MOTS in $(M,g,K)$ and
has zero Dirichlet boundary data at $\del^+ M$. In fact, we chose the
assumptions on the outer boundary $\del^+ M$ so that we can prescribe
more general Dirichlet data there. The focus here lies on the blow-up in
the interior, so that we do not investigate the optimal conditions for
$\del^+M$.
\begin{figure}
  \centering
  \resizebox{.5\linewidth}{!}{\input{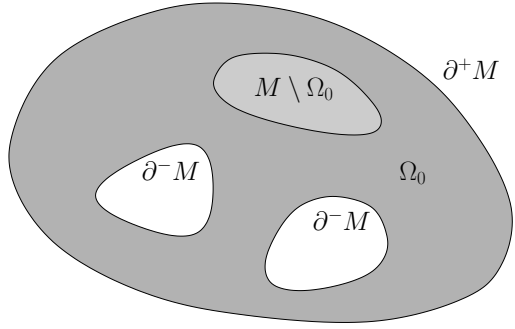}}
  \caption{The situation in Theorem~\ref{thm:blowup}. All of the
    shaded region belongs to $M$, whereas $f$ is only defined in $\Omega_0$.}
\end{figure}
\begin{theorem}
  \label{thm:blowup}
  If $(M,g,K)$ be an initial data set with $\del M = \del^- M \cup
  \del^+ M$ such that $\del^-M$ is an outermost MOTS,
  $\theta^+[\del^+M]>0$ and $\theta^-[\del^+M]<0$, then there exists
  an open set $\Omega_0\subset M$ and a function $f: \Omega_0\to \IR$
  such that
  \begin{enumerate}
  \item $M\setminus\Omega_0$ does not intersect $\del M$,
  \item $\theta^-[\del \Omega_0] =0$ with respect to the normal vector
    pointing into $\Omega_0$, 
  \item $\CJ[f] = 0$,
  \item $N^+ = \graph f \cap M\times \IR^+$ is asymptotic to the cylinder $\del^- M
    \times\IR^+$,
  \item $N^- = \graph f \cap M\times \IR^-$ is asymptotic to the cylinder $\del\Omega_0
    \times\IR^-$, and
  \item $f |_{\del^+ M} = 0$.
  \end{enumerate}
\end{theorem}
For data sets $(M,g,K)$ which do not contain surfaces with
$\theta^-=0$, the above theorem implies the following result.
\begin{corollary}
  \label{thm:blowup_no_mits}
  If $(M,g,K)$ is as in theorem~\ref{thm:blowup}, and in addition
  there are no subsets $\Omega\subset M$ with $\theta^-[\del\Omega]
  =0$ with respect to the normal pointing out of $\Omega$, then there
  exists a function $f:M\to \IR$ such that
  \begin{enumerate}
  \item $\CJ[f] = 0$,
  \item $N = \graph f$ is asymptotic to the cylinder $\del^- M
    \times\IR^+$,
  \item $f |_{\del^+ M} = 0$.
  \end{enumerate}
\end{corollary}
\begin{remark}
  Analogous results hold if $(M,g,K)$ is asymptotically flat with
  appropriate decay of $g$ and $K$ instead of having an outer boundary
  $\del^+ M$. Then the assertion $f |_{\del^+ M} = 0$ in
  theorem~\ref{thm:blowup} has to be replaced by $f(x)\to 0$ as
  $x\to\infty$.
\end{remark}
The proof of theorem~\ref{thm:blowup} is largely based on the tools
developed in \cite{Schoen-Yau:1981} and
\cite{Andersson-Metzger:2007}. Thus we will not include all details
here, but provide a summary, which facts will have to be used.
\begin{proof}
  We will assume that $(M,g,K)$ is embedded into $(M', g', K')$ which
  extends $M$ beyond the boundary $\del^- M$ such that $\del^- M$ lies
  in the interior of $M'$, without further requirements.

  Let $\del^- M = \cup_{i=1}^N \Sigma_i$ where the $\Sigma_i$ are the
  connected components of $\del M$. As $\del M$ is an outermost MOTS,
  each of the $\Sigma_i$ is stable \cite[Corollary
  5.3]{Andersson-Metzger:2007}.

  Following the proof of \cite[Theorem 5.1]{Andersson-Metzger:2007},
  we deform $\del^- M$ to a surface $\Sigma_s$ by pushing the
  components $\Sigma_i$ out of $M$, into the extension $M'$. To
  this end, let $\phi_i>0$ be the principal eigenfunction of the
  stability operator of $\Sigma_i$ and extend the vector field $X_i =
  -\phi_i\nu_i$ to a neighborhood of $\Sigma_i$ in $M'$. Flowing
  $\Sigma_i$ by $X_i$ yields a family of surfaces $\Sigma_i^s$,
  $s\in[0,\eps)$ so that the $\Sigma_i^s$ form a smooth foliation for
  small enough $\eps$ with $\Sigma_i^s \in M' \setminus M$. If
  $\Sigma_i$ is strictly stable then
  \begin{equation*}
    \ddeval{}{s}{s=0}\theta^+[\Sigma_s] = -\lambda\phi < 0,
  \end{equation*}
  where $\lambda$ is the principal eigenvalue of $\Sigma_i$. Thus, for
  small enough $\eps$, we have $\theta^+[\Sigma_i^s] < 0$ for all
  $s\in(0,\eps)$.

  If $\Sigma_i$ has principal eigenvalue $\lambda = 0$, then the 
  $\Sigma_i^s$ satisfy
  \begin{equation*}
    \ddeval{}{s}{s=0}\theta^+[\Sigma_s] = 0.
  \end{equation*}
  In this case it is possible to change the data $K'$ on $\Sigma_i^s$
  as follows
  \begin{equation}
    \label{eq:2}
    \tilde K = K' - \half \psi(s) \gamma_s.
  \end{equation}
  where $\gamma_s$ is the metric on $\Sigma_s$ and $\psi$ is a smooth
  function $\psi:[0,\eps]\to\IR$. The operator $\tilde \theta^+$,
  which means $\theta^+$ computed with respect to the data $\tilde
  K$ instead of $K'$, satisfies
  \begin{equation*}
    \tilde\theta^+[\Sigma_i^s] = \theta^+[\Sigma_i^s] - \psi(s).
  \end{equation*}
  It is clear from equation~(\ref{eq:2}) that $\psi$ can be chosen
  such that $\psi(0) = \psi'(0) = 0$ and
  $\tilde\theta^+[\Sigma_i^s]<0$ for all $s\in(0,\eps)$ provided
  $\eps$ is small enough. Then $\tilde K$ is $C^{1,1}$ when extended
  by $K$ to the rest of $M$.

  Replace each original boundary component $\Sigma_i$ of $M$ by a
  surface $\Sigma_i^\eps$ as constructed above, and replace $K'$ with
  $\tilde K$, such that the following properties are satisfied. Let
  $\tilde M$ denote the manifold with boundary components
  $\Sigma_i^\eps$ resulting from this procedure. Thus we
  construct from $(M,g,K)$ a data set $(\tilde M, g', \tilde K)$ with
  the following properties:
  \begin{enumerate}
  \item $M\subset \tilde M$ with $g'|_M = g$, $\tilde K|_M = K$, and
    $\del^+M = \del^+ \tilde M$,
  \item $\theta^+[\del^- \tilde M] <0$, and    
  \item the region $\tilde M \setminus M$ is foliated by surfaces
    $\Sigma_s$ with $\theta^+(\Sigma_s) <0$.
  \end{enumerate}

  The method developed in section 3.2 in \cite{Andersson-Metzger:2007}
  now allows the modification of the data $(\tilde M, g', \tilde K)$
  to a new data set, which we also denote by $(\tilde M, \tilde g,
  \tilde K)$, although $\tilde K$ changes in this step. This data set
  has the following properties
  \begin{enumerate}
  \item $M\subset \tilde M$ with $g'|_M = g$, $\tilde K|_M = K$, and
    $\del^+M = \del^+ \tilde M$,
  \item $\theta^+[\del^- \tilde M] <0$,
  \item $H[\del^-\tilde M] > 0$ where $H$ is the mean curvature of
    $\del^-M$ with respect to the normal pointing out of $\del^-
    \tilde M$,
  \item the region $\tilde M \setminus M$ is foliated by surfaces
    $\Sigma_s$ with $\theta^+(\Sigma_s) <0$.
  \end{enumerate}

  By section 3.3 in \cite{Andersson-Metzger:2007} this enables us to
  solve the boundary value problem
  \begin{equation}
    \label{eq:3}
    \begin{cases}
      \CJ[f_\tau] = \tau f_\tau & \text{in}\ \tilde M \\
      f_\tau      = \frac{\delta}{2\tau} & \text{on}\ \del^-\tilde M
      \\
      f_\tau      = 0 & \text{on} \del^+\tilde M
    \end{cases}    
  \end{equation}
  where $\delta$ is a lower bound for $H$ on $\del^-M$. The
  solvability of this equation follows, provided an estimate for the
  gradient at the boundary can be found. The barrier construction at
  $\del^-\tilde M$ was carried out in detail in
  \cite{Andersson-Metzger:2007}, whereas the barrier construction at
  $\del^+ \tilde M$ is standard due to the stronger requirement that
  $\theta^+[\del^+M] >0$ and $\theta^-[\del^+M] <0$.
  
  The solution $f_\tau$ to equation~(\ref{eq:3}) satisfies an estimate
  of the form
  \begin{equation}
    \label{eq:4}
    \sup_{\tilde M} |f_\tau| + \sup_{\tilde M} |\nabla f_\tau| \leq \frac{C}{\tau},
  \end{equation}
  where $C$ is a constant depending only on the data $(\tilde M,\tilde
  g, \tilde K)$ but not on $\tau$.

  The gradient estimate implies in particular that there exists an
  $\eps>0$ independent of $\tau$ such that
  \begin{equation*}
    f_\tau(x) \geq \tfrac{\delta}{4\tau} \quad \forall x\ \text{with}\
    \dist(x,\del^-\tilde M)
  \end{equation*}
  The graphs $N_\tau$ have uniformly bounded curvature in $\tilde M
  \times \IR$ away from the boundary. This allows the extraction of a
  sequence $\tau_i\to 0$ such that the $N_{\tau_i}$ converge to a
  manifold $N$, cf.~\cite[Proposition 3.8]{Andersson-Metzger:2007},
  \cite[Section 4]{Schoen-Yau:1981}. This convergence determines three
  subsets of $\tilde M$:
  \begin{equation*}
    \begin{split}
      \Omega_-
      &:=
      \{ x\in M : f_{\tau_i}(x) \to -\infty\ \text{as}\ i \to\infty\},
      \\[5pt]
      \Omega_0
      &:=
      \{ x\in M : \limsup_{i\to\infty} |f_{\tau_i}(x)| < \infty \},
      \\
      \Omega_+
      &:=
      \{ x\in M : f_{\tau_i}(x) \to \infty\ \text{as}\ i \to\infty\}.
    \end{split}
  \end{equation*}
  From the fact that the $f_\tau$ blow up near $\del^-\tilde M$, we
  have that $\Omega_+\neq\emptyset$ and $\Omega_+$ contains a
  neighborhood of $\del^-\tilde \Omega$. As already noted in
  \cite{Schoen-Yau:1981} $\del\Omega_+ \setminus \del\tilde M$
  consists of MOTS. As the region $\tilde M \setminus M$ is foliated
  by surfaces with $\theta^+<0$, we must have that $\Omega_+\supset
  (\tilde M \setminus M)$ and hence $\del\Omega_+$ is a MOTS in
  $M$. As $\del^- M$ was assumed to be an outermost MOTS in $M$, we
  conclude that the closure of $\Omega_+$ is $\tilde M \setminus M$.

  The barriers near $\del^+ M$ are so that they imply that the
  $f_\tau$ are uniformly bounded near $\del^+ M$. Thus $\Omega_0$
  contains a neighborhood of $\del^+ M$ and $\Omega_0\subset M$.

  The limit manifold $N$ over $\Omega_0$ is a graph satisfying
  $\CJ[f_\tau]=0$, and has the desired asymptotics.  
\end{proof}
We will now discuss a geometric condition to assert that the resulting
graph is non-singular on $M$, ie.\ $M=\Omega_0$ in
theorem~\ref{thm:blowup}.
\begin{theorem}
  \label{thm:nonpos_mc}
  Let $(M,g,K)$ be as in theorem~\ref{thm:blowup} with $\tr K \leq
  0$. Then in the assertion of theorem~\ref{thm:blowup} we have that
  $\Omega_0 = M$, that is $f$ is defined on $M$ and has no other
  blow-up than near $\del^-M$.
\end{theorem}
\begin{proof}
  This follows from a simple argument using the maximum principle. Let
  $f_\tau$ be a solution to the regularized problem
  \begin{equation}
    \label{eq:5}
    \CH[f_\tau] - \CP[f_\tau] - \tau f_\tau = 0
  \end{equation}
  in $\tilde M$, as in the proof of theorem~\ref{thm:blowup}. We claim
  that $f_\tau$ can not have a negative minimum in the region where
  the data is unmodified. Assume that $x\in M$ is such a
  minimum. There we have $\CH[f_\tau]\geq 0$, and since $\graph f$ is
  horizontal at $x$ we have that
  \begin{equation*}
    \CP[f_\tau] = \tr K \leq 0.
  \end{equation*}
  thus the right hand side of \eqref{eq:5} is non-negative, whereas
  $\tau f_\tau$ is assumed to be negative, a contradiction.

  Since we know that in the limit $\tau\to 0$, the functions $f_\tau$
  must blow-up in the modified region which lies in $\Omega_+$, we
  infer a lower bound for $f_\tau$ from the above argument. Thus
  $\Omega_0 = M$ as claimed.
\end{proof}
  

%
\section{Asymptotic behavior}
\label{sec:asymptotic-behavior}
Here, we discuss a refinement of \cite[Corollary 2]{Schoen-Yau:1981},
which says that $N = \graph f$ converges uniformly in $C^2$ to the
cylinder $\del^- M \times \IR$ for large values of $f$. A barrier
construction allows us to determine the asymptotics of this
convergence. Before we present our result, recall the statement of
\cite[Corollary 2]{Schoen-Yau:1981}:
\begin{theorem}
  \label{thm:sideways}
  Let $N = \graph f$ the manifold constructed in the proof of
  theorem~\ref{thm:blowup} and let $\Sigma$ be a connected component
  of $\del^-M$. Let $U$ be a neighborhood of $\Sigma$ with positive
  distance to $\del^-M \setminus\Sigma$.

  Then for all $\eps>0$ there exists $\bar z = \bar z(\eps)$,
  depending also on the geometry of $(M,g,K)$, such that $N \cap U
  \times [\bar z,\infty)$ can be written as the graph of a function
  $u$ over $C_{\bar z}:= \Sigma \times [\bar z,\infty)$, so that
  \begin{equation*}
    |u(p,z)| + |\Nabla{C_{\bar z}} u(p,z) | + |\Nabla{C_{\bar z}}^2 u(p,z)| < \eps.
  \end{equation*}
  for all $(p,z)\in C_{\bar z}$. Here, $\Nabla{C_{\bar z}}$ denotes
  covariant differentiation along $C_{\bar z}$.
\end{theorem}
If $\Sigma$ is strictly stable, we can in fact say more about $u$.
\begin{theorem}
  \label{thm:exponential}
  Assume the situation of theorem~\ref{thm:sideways}. If in addition $\Sigma$ is
  strictly stable with principal eigenvalue $\lambda >0$, we have that
  for all $\delta < \sqrt{\lambda}$ there exists $c=c(\delta)$ depending only
  on the data $(M,g,K)$ and $\delta$ such that
  \begin{equation*}
    |u(p,z)| + |\Nabla{C_{\bar z}} u(p,z)| + |\Nabla{C_{\bar z}}^2 u(p,z)| \leq c\exp(-\delta z).
  \end{equation*}
\end{theorem}
\begin{proof}
  Denote by $\beta>0$ the eigenfunction to the principal eigenvalue
  $\lambda$ on $\Sigma$ normalized such that $\max_\Sigma\beta =
  1$. We denote by $\nu$ the normal vector field of $\Sigma$ pointing
  into $M$. Consider the map
  \begin{equation}
    \label{eq:7}
    \Psi : \Sigma\times [0,\bar s] \to M : (p,s) \mapsto \exp^M_p(s\beta\nu).
  \end{equation}
  Given $\eps>0$ we can choose $\bar s>0$ small enough such that the
  surfaces $\Sigma_s=\Psi(\Sigma,s)$ with $s\in[0,\bar s]$ form a
  local foliation near $\Sigma$ with lapse $\beta$ such that
  \begin{equation*}
    \theta^+[\Sigma_s] \geq \lambda(1-\eps)\beta s.
  \end{equation*}
  Denote the region swiped out by these $\Sigma_s$ by $U_{\bar s}$. Note that
  $\del U_{\bar s} = \Sigma \cup \Sigma_{\bar s}$ and $\dist(\Sigma_{\bar
    s},\Sigma)>0$. We can assume that $\dist(\Sigma_{\bar s},\del
  M)>0$. On $U_{\bar s}$ we consider functions $w$ of the form $w=\phi(s)$. For
  such functions Jang's operator can be computed as follows
  \begin{equation*}
    \CJ[w]
    =
    \frac{\phi'}{\beta\sigma}\theta^+
    - \left( 1+ \frac{\phi'}{\beta\sigma}\right) P
    - \sigma^{-2} K(\nu,\nu)
    + \frac{\phi''}{\beta^2\sigma^3},
  \end{equation*}
  where $\sigma^2 = 1 + \beta^{-2}(\phi')^2$, and $\phi'$ denotes the
  derivative of $\phi$ with respect to $s$,
  cf.~\cite{Andersson-Metzger:2007}. The quantities $\theta^+$,
  $K(\nu,\nu)$ and $P$ are computed on the respective $\Sigma_s$.

  Note that with our normalization
  \begin{equation*}
    \sigma^{-2} \leq \beta^{2}|\phi'|^{-2} \leq |\phi'|^{-2}
  \end{equation*}
  and if we assume that $\phi'\geq \mu$ for a large $\mu=\mu(\beta)$
  we have
  \begin{equation*}
    \left| 1 + \frac{\phi'}{\beta\sigma} \right|
    \leq
    2 |\phi'|^{-2}.
  \end{equation*}
  Furthermore,
  \begin{equation*}
    \left|\frac{\phi''}{\beta^2\sigma^3}\right|
    =
    \frac{|\phi''|}{\beta^2 (1+\beta^{-2}\phi'^2)^{3/2}}
    \leq
    \frac{|\phi''|}{\beta^2 (\beta^{-2}\phi'^2)^{3/2}}
    =
    \beta\frac{|\phi''|}{|\phi'|^3}.
  \end{equation*}
  On the other hand, increasing $\mu = \mu(\beta,\eps)$ if necessary,
  we have
  \begin{equation*}
    \frac{|\phi'|}{\beta\sigma}
    \geq
    1 - \eps,
  \end{equation*}
  if $|\phi'|\geq \mu$.

  In combination we find that
  \begin{equation}
    \label{eq:6}
    \CJ[w]
    \leq
    -\lambda(1-\eps)\beta s + \frac{c_1}{|\phi'|^2} + \beta\frac{|\phi''|}{|\phi'|^3},
  \end{equation}
  with $c>0$ depending on $\eps$ and the data $(M,g,K)$, provided
  $|\phi'|\geq \mu$ and $\phi'<0$.

  Choosing $\phi(s) = a \log s$ with $a = (1-\eps)^{-1}
  \lambda^{-1/2}$, we calculate that
  \begin{align*}
    \phi'(s) = \frac{a}{s},
    &\qquad
    \phi''(s) = -\frac{a}{s^2},
    \intertext{so that}
    \frac{1}{|\phi'|^2} = \frac{s^2}{a^2} = (1-\eps)^2\lambda s^2,
    &\qquad
    \frac{\phi''}{|\phi'|^3} = \frac{s}{a^2} = (1-\eps)^2\lambda s.
  \end{align*}
  Thus we can choose $\bar s$ so small that $|\phi'|\geq
  \mu(\beta,\eps)$ and the estimate in~\eqref{eq:6} holds. We can then
  decrease $\bar s$ further, so that $\bar s \leq
  \eps\beta/(c_1(1-\eps))$. This choice makes the right hand side
  of~\eqref{eq:6} non-positive, that is $\CJ[w] \leq 0$.
  Hence, we obtain a super-solution $w$ with
  \begin{equation*}
    \CJ_\tau w \leq 0
  \end{equation*}
  at least where $w \geq 0$, that is near $\Sigma$.

  As $w$ blows up near the horizon, and the $f_\tau$ are bounded
  uniformly in $\tau$ on $\Sigma_{\bar s}$, we can translate $w$ vertically
  to $\bar w = w+b$ with a suitable $b>0$ so that
  \begin{equation*}
    f_\tau|_{\Sigma_{\bar s}} \leq \bar w |_{\Sigma_{\bar s}} 
  \end{equation*}
  for all $\tau>0$. Then the maximum principle implies that
  $f_\tau\leq \bar w$ for all $\tau>0$ in $U_{\bar s}$ and consequently the function
  $f$ constructed in theorem~\ref{thm:blowup} also satisfies $f\leq
  \bar w$.

  Near $\Sigma$, the graph of $\bar w$ can be written as the graph of
  a function $\bar v$ over $\Sigma\times (\bar z,\infty)$ where $v$
  decays exponentially in $z$. This is due to the fact that by the
  assumptions on $\beta$, the parameter $s$ is comparable to the
  distance to $\Sigma$. By the above construction $u\leq v$, where $u$
  is the function from theorem~\ref{thm:sideways}. Thus we find the
  claimed estimate for $u$.

  Getting the desired estimates for the derivatives of $u$ is then a
  standard procedure, but as it is a little work to set the stage, we
  briefly indicate how to proceed.

  We choose coordinates of a neighborhood $\Sigma\times\IR$ in a
  slightly different manner as above. Let $\bar \Psi : \Sigma \times
  (-\eps,\eps)\to M$ be the map
  \begin{equation*}
    \bar\Psi
    :
    \Sigma\times (-\eps,\eps)\times \IR \to M \times\IR
    :
    (x,s,z) \mapsto \big( \exp_x(s\nu), z \big).    
  \end{equation*}
  For a function $h$ on $C_{\bar z}$ we let $\graph_{\bar\Psi} h$ be the set
  \begin{equation*}
    \graph_{\bar\Psi} h = \{ \bar\Psi(x, h(x),z) : (x,z) \in \Sigma\times \IR \}.
  \end{equation*}
  From theorem~\ref{thm:sideways}, it is clear that for large enough
  $\bar z$ the set $N\cap M\times [\bar z,\infty)$ can be written as
  $\graph_{\bar\Psi} h$, where $h$ decays exponentially by the above
  reasoning. We can compute the value of Jang's operator for $h$ as
  follows
  \begin{equation*}
    (\bar H - \bar P)[N] = \CJ h
  \end{equation*}
  where $\CJ$ is a quasi-linear elliptic operator of mean curvature
  type. To be more precise, $\CJ h$ has the form
  \begin{equation}
    \label{eq:8}
    \begin{split}
      \CJ h
      &=
      \del_z^2 h
      + \gamma_{h(x,z)}^{ij} \nabla^2_{i,j} h
      - 2 \gamma_{h(x,z)}^{ij} \del_i(h) K(\del_s, \del_j)
      - \theta^+[\Sigma_{h(x,z)}]\\
      &\qquad + Q(h,\Nabla{C_{\bar z}}h,
      \Nabla{C_{\bar z}}^2 h)
    \end{split}
  \end{equation}
  where $\gamma_s$ is the metric on $\Sigma_s$ and $Q$ is of the form
  \begin{equation*}
    Q(h,\Nabla{C_{\bar z}}h,\Nabla{C_{\bar z}}^2 h)
    =
    h * \Nabla{C_{\bar z}}
    + \Nabla{C_{\bar z}} * \Nabla{C_{\bar z}}
    + \Nabla{C_{\bar z}} * \Nabla{C_{\bar z}} * \Nabla{C_{\bar z}}^2 h
  \end{equation*}
  where $*$ denotes some contraction with a bounded
  tensor. Furthermore, the vectors $\del_i$, $i=1,2$ denote directions
  tangential to $\Sigma$ and $\del_z$ the direction along the
  $\IR$-factor in $C_{\bar z}$.

  By freezing coefficients, we therefore conclude
  that $h$ satisfies a linear, uniformly elliptic equation of the form
  \begin{equation*}
    a^{ij} \del_i\del_j h + \la b, \Nabla{C_{\bar z}} h \ra - \theta^+[\Sigma_{h(x,z)}] = 0.
  \end{equation*}
  By construction we have that $|\theta^+[\Sigma_s]| \leq \kappa s$
  for some fixed $\kappa$. Thus $\theta^+[\Sigma_{h(x,z)}]$ decays
  exponentially in $z$.

  Now we are in the position to use standard interior estimates for
  linear elliptic equations to conclude the decay of higher
  derivatives of $h$. This decay translates back into the decay of the
  first and second derivatives of $u$ as the coordinate transformation
  is smooth and controlled by the geometry of $(M,g,K)$.
\end{proof}
\begin{remark}
  If $\Sigma$ is not strictly stable, but has positive $k$-th
  variation, we find that the foliation near $\Sigma$ satisfies
  $\theta^+[\Sigma_s] \geq \kappa s^k$. Then a function of the form
  $\phi(s) = a s^{-p}$ with large $a$ and $p = \frac{k-1}{2}$ yields
  a super-solution. This super-solution can be used to prove that
  $|u| \leq C z^{2/(1-k)}$ as above.
\end{remark}
We can get even more information about the decay rate. A closer look
at equation~\eqref{eq:8} yields that the expression for $\CJ h$ on
$C_{\bar z}$ can also be written as follows
\begin{equation*}
  \CJ h = (\del_z^2 - L_M) h + Q'(h, \Nabla{C_{\bar z}} h, \Nabla{C_{\bar z}}^2 h),
\end{equation*}
since
\begin{align*}
  &\theta^+_s
  =
  s L_M 1 + \CO(s^2),
  \\
  &\gamma_{h(x,z)}^{ij} \nabla^2_{i,j} h
  =
  \Delta h + Q_1(h,\nabla h,\nabla^2 h),
  \intertext{and}
  &\gamma_{h(x,z)}^{ij} \del_i h K(\del_s, \del_j)
  =
  S(\nabla h) + Q_2(h,\nabla h),
\end{align*}
where the differential operators $\nabla$ and $\Delta$ are with
respect to $\Sigma$. Then note that
\begin{equation*}
  L_M h = h L_M 1 - \Delta h + 2 S(\nabla h).
\end{equation*}
Further investigation of the structure of $Q'$
yields that
\begin{equation*}
  |Q'(h, \Nabla{C_{\bar z}} h, \Nabla{C_{\bar z}}^2 h)|
  \leq
  C \big(|h|^2 + |\Nabla{C_{\bar z}} h|^2 + |h| |\Nabla{C_{\bar z}}^2
  h| + |\Nabla{C_{\bar z}} h|^2 |\Nabla{C_{\bar z}}^2
  h| \big),
\end{equation*}
so that in view of the differential Harnack estimate $|\Nabla{C_{\bar
    z}} h| \leq c |h|$ for positive solutions of linear elliptic equations we
have that in fact
\begin{equation*}
  |Q'(h,\Nabla{C_{\bar z}} h,\Nabla{C_{\bar z}}^2 h)|
  \leq
  c |h|\big(|h| + |\Nabla{C_{\bar z}} h| + |\Nabla{C_{\bar z}}^2 h|\big),
\end{equation*}
provided $|h|\leq C$.  By projecting the equation $\CJ h = 0$ to the
one-dimensional eigenspace of $L_M$ it is now a somewhat standard ODE
argument to show the following result.
\begin{theorem}
  \label{thm:rigidity}
  Under the assumptions of theorem~\ref{thm:exponential} there are no
  solutions $h : \Sigma\times [0,\infty) \to \IR$ to the equation
  \begin{equation}
    \label{eq:9}
    \CJ h = 0
  \end{equation}
  with decay 
  \begin{equation*}
    |h(p,z)| + |\Nabla{C_{\bar z}} h(p,z)| + |\Nabla{C_{\bar z}} h(p,z)|
    \leq
    C \exp(-\delta z)
  \end{equation*}
  such that $\delta > \sqrt{\lambda}$ and $h>0$.
\end{theorem}
\begin{proof}
  Assume that $h>0$ is such a solution. We derive a contradiction as
  follows. Let $\lambda$ be the principal eigenvalue and $\phi$ be the
  corresponding eigenfunction of $L_M$ as before. Let $L_M^*$ be the
  (formal) adjoint of $L_M$ on $L^2(\Sigma)$ and denote by $\phi^*>0$
  its principal eigenfunction, normalized such that $\int_\Sigma
  \phi\phi^* \dmu = 1$. Then the operator
  \begin{equation*}
    P u
    =
    \left(\int_\Sigma \phi^* u \dmu\right)\ \phi
  \end{equation*}
  is a projection onto the eigenspace spanned by $\phi$ and moreover
  commutes with $L_M$. We interpret $h(z)$ as a family of
  functions on $\Sigma$, that is $h(z)(p) = h(p,z)$ for
  $p\in\Sigma$. Choose $\alpha(z)$ such that
  \begin{equation*}
    P h(z) = \alpha(z) \phi.
  \end{equation*}
  and $\beta(z)$ accordingly,
  \begin{equation*}
    \beta(z)\phi
    =
    P \big(Q'(h(\cdot,z),\Nabla{C_{\bar z}} h(\cdot,z), \Nabla{C_{\bar
        z}}^2 h(\cdot,p) )\big).
  \end{equation*}
  Then equation~\eqref{eq:9} and the fact that $P$ commutes with $L_M$
  and $\del_z$ imply
  \begin{equation*}
    \alpha''(z) - \lambda \alpha(z) = \beta(z).
  \end{equation*}
  Using $\phi^* >0$ and $h>0$ yields $\alpha(z)>0$ and we can
  furthermore estimate that
  \begin{equation*}
    \begin{split}
      |\beta(z)|
      &
      \leq
      c \int_\Sigma \phi^* |h(p,z)| \big(|h(p,z)| + |\nabla h(p,z)| +
      |\nabla^2 h(p,z)|\big) \dmu
      \\
      &
      \leq
      c\exp(-\delta z) \int_\Sigma \phi^* |h(p,z)|\dmu
      \\
      &
      \leq
      c \exp(-\delta z) \alpha(z).
    \end{split}
  \end{equation*}
  Thus, we conclude that on $[\tilde z, 0)$ the function
  $\alpha>0$ satisfies a differential inequality of the form
  \begin{equation*}
    \alpha''(z) - \lambda \alpha  \leq \eps \alpha.
  \end{equation*}
  where $\eps>0$ can be chosen arbitrarily small by choosing $\tilde
  z$ large enough. If $\sqrt{\lambda + \eps} < \delta$ this ODE has no
  solutions with decay $\exp(-\delta z)$ other than the trivial
  solution. Thus $\alpha \equiv 0$ and we arrive at the desired
  contradiction.
\end{proof}


%
\section*{Acknowledgments}
The author thanks the Mittag-Leffler-Institute, Djursholm, Sweden for
hospitality and support during the program \emph{Geometry, Analysis,
  and General Relativity} in Fall 2008.


%
\bibliographystyle{amsalpha}
\bibliography{../extern/references}

\providecommand{\bysame}{\leavevmode\hbox to3em{\hrulefill}\thinspace}
\providecommand{\MR}{\relax\ifhmode\unskip\space\fi MR }
\providecommand{\MRhref}[2]{%
  \href{http://www.ams.org/mathscinet-getitem?mr=#1}{#2}
}
\providecommand{\href}[2]{#2}
\begin{thebibliography}{M{\'O}M04}

\bibitem[AM07]{Andersson-Metzger:2007}
L.~Andersson and J.~Metzger, \emph{The area of horizons and the trapped
  region}, arXiv:0708.4252 [gr-qc], 2007.

\bibitem[AMS05]{Andersson-Mars-Simon:2005}
L.~Andersson, M.~Mars, and W.~Simon, \emph{Local existence of dynamical and
  trapping horizons}, Phys. Rev. Lett. \textbf{95} (2005), 111102,
  arXiv:gr-qc/0506013.

\bibitem[AMS07]{andersson-mars-simon:2007}
\bysame, \emph{Stability of marginally outer trapped surfaces and existence of
  marginally outer trapped tubes}, arXiv:0704.2889 [gr-qc], 2007.

\bibitem[Bra01]{Bray:2001}
H.~L. Bray, \emph{Proof of the {R}iemannian {P}enrose inequality using the
  positive mass theorem}, J. Differential Geom. \textbf{59} (2001), no.~2,
  177--267.

\bibitem[HI01]{Huisken-Ilmanen:2001}
G.~Huisken and T.~Ilmanen, \emph{The inverse mean curvature flow and the
  {R}iemannian {P}enrose inequality}, J. Differential Geom. \textbf{59} (2001),
  no.~3, 353--437.

\bibitem[Jan78]{Jang:1978}
P.~S. Jang, \emph{On the positivity of energy in general relativity}, J. Math.
  Phys. \textbf{19} (1978), 1152--1155.

\bibitem[Khu09]{khuri:2009}
M.~Khuri, \emph{A penrose-like inequality for general initial data sets}, Comm.
  Math. Phys. \textbf{290} (2009), no.~2, 779--788.

\bibitem[M{\'O}M04]{Malec-OMurchadha:2004}
Edward Malec and Niall {\'O}~Murchadha, \emph{The {J}ang equation, apparent
  horizons and the {P}enrose inequality}, Classical Quantum Gravity \textbf{21}
  (2004), no.~24, 5777--5787. \MR{MR2107339 (2005i:83025)}

\bibitem[Sch04]{Schoen:2004}
R.~Schoen, {Talk given at the Miami Waves conference}, January 2004.

\bibitem[SY81]{Schoen-Yau:1981}
R.~Schoen and S.-T. Yau, \emph{Proof of the positive mass theorem. {II}}, Comm.
  Math. Phys. \textbf{79} (1981), no.~2, 231--260.

\end{thebibliography}
\end{document}